\documentclass[a4paper,UKenglish,cleveref, autoref, thm-restate,authorcolumns]{lipics-v2019}

\usepackage{microtype}
\usepackage{amsthm,amsmath}
\usepackage{thmtools}
\usepackage{amssymb}
\usepackage{hyperref}
\usepackage{graphicx,inputenc,transparent}

\usepackage[utf8]{inputenc}

\usepackage{lipsum}       
\usepackage{changepage}   

\usepackage{xspace,cite,dsfont}

\usepackage{xcolor}

\usepackage[ruled, vlined,boxed,commentsnumbered,linesnumbered]{algorithm2e}
\definecolor{dark-red}{rgb}{0.4,0.15,0.15}
\definecolor{dark-blue}{rgb}{0.15,0.15,0.4}
\definecolor{medium-blue}{rgb}{0,0,0.5}
\definecolor{gray}{rgb}{0.5,0.5,0.5}
\definecolor{color-Ig}{rgb}{0.15,0.7,0.15}

\hypersetup{
    colorlinks, linkcolor={dark-red},
    citecolor={dark-blue}, urlcolor={medium-blue}
}

\newcommand{\NP}{\ensuremath{\mathsf{NP}}\xspace}
\renewcommand{\P}{\ensuremath{\mathsf{P}}\xspace}

\newcommand{\Acal}{\mathcal{A}}

\newcommand{\ETH}{{\sf ETH}\xspace}
\newcommand{\FPT}{{\sf FPT}\xspace}
\newcommand{\XP}{{\sf XP}\xspace}

\renewcommand{\P}{{\sf P}\xspace}

\renewcommand{\boldmath}{}

\newcommand{\true}{{\sf true}\xspace}
\newcommand{\false}{{\sf false}\xspace}

\theoremstyle{plain}

\newtheorem{Claim}[theorem]{Claim}

\newenvironment{proof1}[1][Proof of Theorem~\ref{safefpt}]{\begin{proof}[#1]}{\end{proof}}

\newenvironment{proof2}[1][Proof of Theorem~\ref{flowfpt}]{\begin{proof}[#1]}{\end{proof}}

\newenvironment{proof3}[1][Proof of Theorem~\ref{treefpt}]{\begin{proof}[#1]}{\end{proof}}

\newenvironment{proof4}[1][Proof of Theorem~\ref{theo:lower-graph}]{\begin{proof}[#1]}{\end{proof}}

\hypersetup{
colorlinks = true,
linkcolor = black!30!blue,
citecolor = black!30!green
}

\title{FPT algorithms for packing $k$-safe spanning rooted sub(di)graphs}

\titlerunning{FPT algorithms for packing $k$-safe spanning rooted sub(di)graphs} 

\author{St\'ephane Bessy}{LIRMM, Universit\'e de Montpellier, CNRS, Montpellier, France}{stephane.bessy@lirmm.fr}{https://orcid.org/0000-0001-7130-4990}{{\footnotesize DIGRAPHS (ANR-19-CE48-0013-02).}}

\author{Florian H\"orsch}{Université Grenoble Alpes, Grenoble INP, CNRS, G-SCOP, Grenoble, France}{florian.hoersch@grenoble-inp.fr}{}{}

\author{Ana Karolinna Maia}{Departamento de Computa\c c\~ao, Universidade Federal do Cear\'a, Fortaleza, Brazil}{karolmaia@ufc.br}{https://orcid.org/0000-0002-9027-7948}{{\footnotesize FUNCAP Pronem 4543945/2016 and  CAPES/STIC-AmSud~88881.197438/2018-01.}}

\author{Dieter Rautenbach}{Institute of Optimization and Operations Research, Ulm University, Germany}{dieter.rautenbach@uni-ulm.de}{https://orcid.org/0000-0002-7214-042X}{}

\author{Ignasi Sau}{LIRMM, Universit\'e de Montpellier, CNRS, Montpellier, France}{ignasi.sau@lirmm.fr}{https://orcid.org/0000-0002-8981-9287}{{\footnotesize DEMOGRAPH (ANR-16-CE40-0028), ESIGMA (ANR-17-CE23-0010), ELIT (ANR-20-CE48-0008-01), and French-German Collaboration ANR/DFG Project UTMA (ANR-20-CE92-0027).}}

\authorrunning{S. Bessy, F. Hoersch, A. K. Maia, D. Rautenbach, and I. Sau} 

\Copyright{St\'ephane Bessy, Florian Hoersch, Ana Karolinna Maia, Dieter Rautenbach, and Ignasi Sau} 



\category{} 

\relatedversion{} 

\supplement{}

\acknowledgements{}

\nolinenumbers 

\hideLIPIcs  

\EventNoEds{2}
\EventLongTitle{}
\EventShortTitle{arXiv preprint}
\EventAcronym{}
\EventYear{2020}
\EventDate{August 24--28, 2020}
\EventLocation{Prague, Czech Republic}
\EventLogo{}
\SeriesVolume{}
\ArticleNo{1}

\begin{document}

\maketitle


\begin{abstract}
We study three problems introduced by Bang-Jensen and Yeo~[Theor. Comput. Sci. 2015] and by Bang-Jensen, Havet, and Yeo~[Discret. Appl. Math. 2016] about finding disjoint ``balanced'' spanning rooted substructures in graphs and digraphs, which generalize classic packing problems. Namely, given a positive integer $k$, a digraph $D=(V,A)$, and a root $r \in V$, we
consider the problem of finding two arc-disjoint $k$-safe spanning
$r$-arborescences and the problem of finding two arc-disjoint
$(r,k)$-flow branchings. We show that both these problems are \FPT
with  parameter $k$, improving on existing \XP algorithms.  The latter of these results answers a question of Bang-Jensen, Havet, and Yeo~[Discret. Appl. Math. 2016]. Further, given an
integer $k$, a graph $G=(V,E)$, and $r \in V$, we consider the
problem of finding two arc-disjoint $(r,k)$-safe spanning trees. We
show that this problem is also \FPT with  parameter
$k$, again improving on a previous \XP algorithm.
Our main technical contribution is to prove that the existence of such spanning substructures is equivalent to the existence of substructures with size and maximum (out-)degree both bounded by a (linear or quadratic) function of $k$, which may be of independent interest.
\end{abstract}


\section{Introduction}
\label{sec:intro}

This article deals with finding certain disjoint substructures in
graphs and digraphs. Throughout the article, when given a graph or a digraph, we use {\boldmath$n$} for its number of vertices.


All graphs and digraphs considered in this paper are loopless, but may have multiple edges or arcs. Given a graph $G=(V,E)$, we say that an edge $e=uv$ is {\it
  incident} to $u$ and $v$. For some $X \subseteq V$, we denote by
{\boldmath $d_G(X)$} the number of edges that are incident to exactly
one vertex in $X$. For some $v \in V$, we use {\boldmath$N_G(v)$} for
the set of vertices $w$ such that there is an edge between $v$ and
$w$. A {\it rooted graph} is a graph $G=(V+r,E)$ with a special vertex
$r$ called the \emph{root}. A vertex $v \in V$ with $d_G(v)=1$ is called a {\it leaf} of $G$.


Given a digraph $D=(V,A)$ and some $X \subseteq V$, we use {\boldmath
  $\delta_D^-(X)$} (resp. {\boldmath $\delta_D^+(X)$}) for the set of
arcs entering (resp. leaving) $X$. We use {\boldmath $d_D^-(X)$}
(resp. {\boldmath $d_D^+(X)$}) for $|\delta_D^-(X)|$
(resp. $|\delta_D^+(X)|$). For a single vertex $v$, we abbreviate
$\delta_D^+(\{v\})$ (resp. $\delta_D^-(\{v\}),d_D^+(\{v\}),d_D^-(\{v\})$) to
{\boldmath $\delta_D^+(v)$} (resp. {\boldmath
  $\delta_D^-(v),d_D^+(v),d_D^-(v)$}). We call $d_D^-(v)$ (resp. $d_D^+(v)$) the
{\it in-degree} (resp. {\it out-degree)} of $v$. We use
{\boldmath$N_D^+(v)$} for the set of vertices $w$ such that there is
an arc from $v$ to $w$. A vertex $v \in V$ with $d^+(v)=0$ is called a {\it sink} of $D$. Subscripts may be omitted when they are clear from the context. The {\it underlying graph} of $D$ is obtained
by replacing all arcs of $A$ by an edge between the same two vertices. A {\it rooted digraph} is a digraph $D=(V+r,A)$ with a special vertex
$r$, called the \emph{root}, whose in-degree is 0.



\medskip

\noindent \textbf{Packing $k$-safe spanning $r$-arborescences}. The first objects we deal with are called arborescences. We remark here that the term \emph{out-branching} is used in \cite{bhy} to describe the same object.
An {\it $r$-arborescence} is a directed graph $X=(V+r,A)$ such that the underlying graph of $X$ is a  tree, the in-degree of
$r$ is $0$ in $X$, and the in-degree of all other vertices  is $1$ in
$X$. We say that $r$ is the root of $X$. Observe that every $v \in V$ is the root of a unique
subarborescence of $X$. We denote this subarborescence by
{\boldmath$B_X^v$}. Given a digraph $D$, an $r$-arborescence $X$ that is a subdigraph of $D$ is {\it spanning} in $D$ if it has the same vertex set as $D$. The following is a fundamental result in digraph theory.

\begin{theorem}[Edmonds~\cite{e}]
\label{basic1}
Let $D=(V+r,A)$ be a rooted digraph and $k$ a positive integer. There
exists a set of $k$ arc-disjoint spanning $r$-arborescences in $D$ if
and only if $d_A^{-}(X)\geq k \text{ for all } \emptyset \neq
X\subseteq V.$
\end{theorem}

A number of alternative proofs of \autoref{basic1} have been found,
several of which are algorithmic and yield polynomial-time algorithms to find the desired arc-disjoint spanning $r$-arborescences, if they exist~\cite{e4,l}.

This naturally raises the question whether we also can efficiently find
spanning arborescences satisfying some extra properties. This consideration, as well as practical applications concerning protection against arc failures,
motivated Bang-Jensen and Yeo~\cite{balanced}
 to introduce the notion of $k$-safe spanning
arborescences. An $r$-arborescence $X=(V+r,A)$ is called {\it
  $k$-safe} if $n-|V(B_X^v)|\geq k$ for all $v\in
V$. Notice that it is enough that only the
  out-neighbours of $r$ satisfy this latter condition for $X$ to be
  $k$-safe. They proved the following negative result showing that in
general not even the problem of finding a {\sl single} $k$-safe spanning $r$-arborescence
is tractable.
\begin{theorem}[Bang-Jensen and Yeo~\cite{balanced}]
Given a rooted digraph $D=(V+r,A)$,  deciding whether
$D$ has an $(n-k)$-safe spanning $r$-arborescence is \NP-complete for any fixed $k
\geq 3$.
\end{theorem}
 In this light, a characterization in the shape of
 \autoref{basic1} clearly seems out of reach. It remains nevertheless
 interesting to investigate the possibility of finding arc-disjoint
 $k$-safe spanning $r$-arborescences for small values of $k$. This question has
 been dealt with by Bang-Jensen, Havet, and Yeo~\cite{bhy}. On the negative side, they implicitly proved the following result,
which shows that there is little hope to algorithmically find
arc-disjoint $k(n)$-safe spanning $r$-arborescences if $k$ is a function that does
not grow too slowly. While a polynomial-time algorithm for the problem
they consider would not imply $\P=\NP$, it would imply the failure of the\emph{ Exponential Time Hypothesis} (\ETH for short) of Impagliazzo and Paturi~\cite{ETH}, stating that there is an $\varepsilon > 0$ such that there is no algorithm for solving a 3-{\sc Sat} formula
with $\ell$ variables and $m$ clauses in time $2^{\varepsilon \ell} \cdot (\ell+m)^{O(1)}$.

\begin{theorem}[Bang-Jensen, Havet, and Yeo~\cite{bhy}]
  \label{theo:lower-arbo}
Suppose that the \ETH holds, let $\varepsilon >0$ be arbitrary and let
$k:\mathbb{Z}_{\geq 0} \rightarrow \mathbb{Z}_{\geq 0}$ be a function
such that $(\log(n))^{1+\varepsilon}\leq k(n) \leq \frac{n}{2}$ for all
$n>0$. Further, suppose that there exists a constant $C^*$ such that
for all $c \geq C^*$ there exists an $n$ such that $k(n)=c$. Then
there is no algorithm running in time $n^{O(1)}$ for deciding whether a
given rooted digraph $D=(V+r,A)$ has two arc-disjoint $k(n)$-safe spanning
$r$-arborescences.
\end{theorem}

On the positive side, they show that the problem becomes tractable
when fixing the value of $k$. While several results considered in this
article hold for finding an {\sl arbitrary number} of disjoint objects,
we focus on the case where we want to find just {\sl two} of them in order to avoid technicalities. In \autoref{con} we discuss the generalization to more than two objects.
\begin{theorem}[Bang-Jensen, Havet, and Yeo~\cite{bhy}]
\label{safexp}
Deciding whether a given rooted digraph $D=(V+r,A)$
contains two arc-disjoint $k$-safe spanning $r$-arborescences is \XP with parameter~$k$.
\end{theorem}

(See \autoref{prelim:pc}
 for the definition of the classes \XP and \FPT.) Our first contribution is to improve \autoref{safexp} by showing that the problem is \emph{fixed-parameter tractable} (\FPT).

\begin{theorem}\label{safefpt}
Deciding whether a given rooted digraph $D=(V+r,A)$ contains two
arc-disjoint $k$-safe spanning $r$-arborescences is \FPT with
parameter $k$. More precisely, it can be solved in time $2^{O(k^2 \cdot \log k)} \cdot n^{O(1)}$. Further, if they exist, the
two arc-disjoint $k$-safe spanning $r$-arborescences can be computed
within the same running time.
\end{theorem}

\medskip

\noindent \textbf{Packing $(r,k)$-flow branchings}. The second structure we consider, which was introduced by Bang-Jensen and Bessy~\cite{bb}, builds a connection between the
theory of finding arc-disjoint spanning arborescences and flow problems. A {\it flow} in a digraph $D=(V,A)$ is a function ${\mathbf z}:A \rightarrow \mathbb{Z}_{\geq
  0}$. Given a rooted
digraph $X=(V+r,A)$ and a capacity function ${\mathbf c}:A \rightarrow
\mathbb{Z}_{\geq 0}$, an {\it $(r,{\mathbf c})$-branching flow} is a flow ${\mathbf z}:A
\rightarrow \mathbb{Z}_{\geq 0}$ such that ${\mathbf z}(a)\leq {\mathbf c}(a)$ for all
$a\in A$, ${\mathbf z}(\delta^+(r))-{\mathbf z}(\delta^-(r))=n-1$ and
${\mathbf z}(\delta^-(v))-{\mathbf z}(\delta^+(v))=1$ for all $v \in V$. If $X$ admits an
$(r,{\mathbf c})$-branching flow, we say that $D$ is an {\it $(r,{\mathbf c})$-flow
  branching}. If for some positive integer $k$, we have ${\mathbf c}(a)=n-k$ for
all $a \in A$, we speak of an {\it $(r,k)$-branching flow} and an {\it
  $(r,k)$-flow branching}. Given a digraph $D$, an $(r,k)$-branching flow $X$ that is a subdigraph of $D$ is {\it spanning} in $D$ if it has the same vertex set as $D$. If $D$ admits a $k$-safe spanning
  $r$-arborescence, then it is easy to see that $D$ is an $(r,k)$-flow
  branching, but the converse is not necessarily true, as pointed out in~\cite{bhy}.

Bang-Jensen and Bessy~\cite{bb} consider the problem of finding arc-disjoint
spanning flow branchings in a rooted digraph. Among others, they show the
following negative result that makes a characterization in the shape
of  \autoref{basic1} seem out of reach.
\begin{theorem}[Bang-Jensen and Bessy~\cite{bb}]
Given a rooted digraph $D=(V+r,A)$ and a capacity function ${\mathbf c}:A
\rightarrow \mathbb{Z}_{\geq 0}$, it is \NP-complete to decide whether
$D$ contains two arc-disjoint spanning $(r,{\mathbf c})$-flow branchings even if ${\mathbf c}(a)\in
\{1,2\}$ for all $a \in A$.
\end{theorem}
The above result has been strengthened in \cite{bhy}.
\begin{theorem}[Bang-Jensen, Havet, and Yeo~\cite{bhy}]
Given a rooted digraph $D=(V+r,A)$ and a fixed positive integer $k\geq 2$,
it is \NP-complete to decide whether $D$ contains two arc-disjoint spanning
$(r,n-k)$-flow branchings.
\end{theorem}
On the other hand, the problem of finding arc-disjoint spanning $(r,k)$-flow
branchings for some small values of $k$ turns out to be more
tractable. The study of such spanning flow branchings has surprisingly many
similarities with the study of $k$-safe spanning arborescences.

On the negative side, the following result is proven in \cite{bhy}. It
shows that there is little hope to algorithmically find
arc-disjoint spanning $(r,k(n))$-flow branchings if $k$ is a function that does
not grow too slowly. It can be viewed as an analogue of \autoref{theo:lower-arbo} for $(r,k)$-flow branchings.


\begin{theorem}[Bang-Jensen, Havet, and Yeo~\cite{bhy}]
  \label{theo:lower-branching}
Suppose that \ETH holds, let $\varepsilon >0$ be arbitrary and let
$k:\mathbb{Z}_{\geq 0} \rightarrow \mathbb{Z}_{\geq 0}$ be a function
such that $(\log(n))^{1+\varepsilon}\leq k(n) \leq \frac{n}{2}$ for all
$n>0$. Further, suppose that there exists a constant $C^*$ such
that for all $c \geq C^*$ there exists an $n$ such that $k(n)=c$. Then
there is no algorithm running in time $n^{O(1)}$ for deciding whether a
given rooted digraph $D=(V+r,A)$ has two arc-disjoint spanning $(r,k(n))$-flow
branchings.
\end{theorem}

On the positive side, Bang-Jensen and Bessy~\cite{bb} showed that the case $k=1$ can be solved in polynomial time. This result was again generalized by
Bang-Jensen, Havet, and Yeo~\cite{bhy}, who proved that
 the problem can be solved in polynomial time for every fixed  value of $k \geq 1$. The following result  can be viewed as an analogue of \autoref{safexp} for $(r,k)$-flow branchings.

\begin{theorem}[Bang-Jensen, Havet, and Yeo~\cite{bhy}]
\label{flowxp}
Deciding whether a given rooted digraph $D=(V+r,A)$ contains two
arc-disjoint spanning $(r,k)$-flow branchings is \XP with
parameter~$k$.
\end{theorem}

The authors of \cite{bhy} ask whether the above problem is \FPT. Our second contribution is an affirmative answer to this question.

\begin{theorem}\label{flowfpt}
Deciding whether a given rooted digraph $D=(V+r,A)$ contains two
arc-disjoint spanning $(r,k)$-flow branchings is \FPT with
parameter $k$. More precisely, it can be solved in time $2^{O(k^2 \cdot \log k)} \cdot n^{O(1)}$. Further, if they exist, the two arc-disjoint spanning $(r,k)$-flow branchings can be computed within the same running time.
\end{theorem}


\medskip

\noindent \textbf{Packing $(r,k)$-safe spanning trees}. Finally, we consider a similar problem in undirected graphs that has
also been introduced in \cite{bhy}. Given a graph $G=(V,E)$, a {\it
  spanning tree} is a subgraph $T$ of $G$ that is a tree with
$V(T)=V$. The theory of finding edge-disjoint spanning trees is also
pretty rich. The most fundamental result is the following one.

\begin{theorem}[Tutte~\cite{t}]
\label{tutte}
Let $G=(V,E)$ be a graph and $k$ a positive integer. Then $G$ has $k$
edge-disjoint spanning trees if and only if $\sum_{i=1}^qd_G(V_i)\geq
2k(q-1)$ for every partition $\{V_1,\ldots,V_q\}$ of $V$.
\end{theorem}

An algorithmic proof of \autoref{tutte}, yielding a polynomial-time algorithm, can be found in
\cite{book}. Again, we may wish to also find edge-disjoint
spanning trees satisfying certain extra properties. Given a rooted
tree $T=(V+r,E)$ and some $v \in V$, we use {\boldmath $C_T^v$} for
the subgraph of $T-v$ that arises from deleting the component of $T-v$
containg $r$. We say that $T$ is {\it $(r,k)$-safe} if for every $v
\in V$, we have $|V-V(C_T^v)|\geq k$.

\begin{theorem}[Bang-Jensen, Havet, and Yeo~\cite{bhy}]
\label{treexp}
Deciding whether a given rooted graph $G=(V+r,E)$
contains two edge-disjoint $(r,k)$-safe spanning trees is \XP with
 parameter $k$.
\end{theorem}

Again, we improve \autoref{treexp} as follows.

\begin{theorem}\label{treefpt}
Deciding whether a given rooted graph $G=(V+r,E)$
contains two edge-disjoint $(r,k)$-safe spanning trees is \FPT with parameter $k$. More precisely, it can be solved in time $2^{O(k^2 \cdot \log k)} \cdot n^{O(1)}$.
Further, if they exist, the two edge-disjoint $(r,k)$-safe spanning trees can be computed within the same running time.
\end{theorem}

Since a hardness result in the spirit of  \autoref{theo:lower-arbo} and \autoref{theo:lower-branching} was not provided in~\cite{bhy}, we fill this gap and prove the following theorem, whose proof is inspired by the one of~\cite[Theorem 5.2]{bhy}. It shows that the problem is hard even if we want to find one single $(r,k)$-safe spanning tree.



\begin{theorem}
  \label{theo:lower-graph}
Let $p \geq 1$ be a fixed positive integer. Deciding whether a
  given rooted graph $G=(V+r,A)$ has $p$ edge-disjoint $(r,k)$-safe
  spanning trees is \NP-complete. Moreover, let $\varepsilon >0$ be arbitrary and let $k:\mathbb{Z}_{\geq 0}
  \rightarrow \mathbb{Z}_{\geq 0}$ be a function such that
  $(\log(n))^{2+\varepsilon}\leq k(n) \leq \frac{n}{2}$ for all
  $n>0$. Further, suppose that there exists a constant $C^*$ such
  that for all $c \geq C^*$ there exists an $n$ such that
  $k(n)=c$. Then, assuming the \ETH, there is no algorithm running in time $n^{O(1)}$ for
  deciding whether a given rooted graph contains $p$ edge-disjoint
  $(r,k)$-safe spanning trees.
  \end{theorem}


\noindent \textbf{Our techniques}. In order to obtain the \FPT algorithms for the three considered problems, we follow a common strategy. In a nutshell, the main ideas used in the \XP algorithms of~\cite{bhy}, and how we manage to improve them to \FPT algorithms, can be summarized as follows. The algorithms of Bang-Jensen, Havet, and Yeo~\cite{bhy} are all based on proving the following general property for each of the considered problems, where all the substructures are rooted:

\smallskip

\begin{adjustwidth}{.7cm}{.7cm}
  A given (di)graph contains the required  substructure ${\cal X}$ (i.e., a pair of disjoint spanning arborescences, flow branchings, or spanning trees) if and only if it contains another type of substructure ${\cal X}'$ {\sl of size bounded by a function of $k$} and such that, if found,  it can be extended to the required substructure ${\cal X}$ {\sl in polynomial time}.
\end{adjustwidth}

\smallskip

\noindent Once the above property is proved, an \XP algorithm follows naturally: generate all candidate substructures ${\cal X}'$ in time $n^{f(k)}$ and, for each of them, try to extend it to a substructure  ${\cal X}$ in polynomial time. Our main contribution is to prove that the above general property is still true if we replace ${\cal X}'$ with another type of substructure ${\cal X}''$ having the crucial property that the candidate substructures ${\cal X}''$ can be all enumerated in time $f(k) \cdot n^{O(1)}$, hence yielding an \FPT algorithm. In order to achieve this, we prove that we can restrict ourselves to objects ${\cal X}''$ whose ``non-sink'' vertices (i.e., those with positive (out-)degree in ${\cal X}''$) have {\sl (out-)degree, in the original (di)graph, bounded by some function of $k$}, namely $O(k)$ or $O(k^2)$.
 Intuitively, this is possible because, given a pair ${\cal X}' = \{X_1,X_2\}$ containing a vertex $v$ of large (out-)degree (as a function of $k$) in, say,  $X_1$, we can safely prune the ``branch'' of $X_1$ hanging from $v$, with the guarantee that it will always be possible to extend the pruned substructure to another substructure of the original type. Note that since the substructures ${\cal X}''$ have size {\sl and} maximum (out-)degree bounded by a function of $k$, we can indeed generate all candidate substructures in time $f(k) \cdot n^{O(1)}$, as required. We now make this informal explanation more concrete. For all technical definitions used in the next paragraphs, see \autoref{pre}.

Let us first focus on the problem of finding two arc-disjoint $k$-safe spanning $r$-arborescences.
In this case, the substructure $\cal{X}'$ is an\emph{ extendable pair of arc-disjoint classic $(r,k)$-kernels}. In \autoref{classksafe}, we restate a result from \cite{bhy} that shows that the existence of this substructure is sufficient for the existence of two arc-disjoint $k$-safe spanning $r$-arborescences. We then introduce compact $(r,k)$-kernels. An \emph{extendable pair of arc-disjoint compact $(r,k)$-kernels} corresponds to the substructure $\cal{X}''$. In \autoref{compactsafe}, we show that the existence of an extendable pair of arc-disjoint compact $(r,k)$-kernels is also sufficient for the existence of two arc-disjoint $k$-safe spanning $r$-arborescences. The proof of \autoref{compactsafe} is our main technical contribution. Having \autoref{compactsafe} at hand, the proof of \autoref{safefpt} is easy.

As for packing $(r,k)$-flow branchings, the substructure ${\cal X}'$ defined in~\cite{bhy} is an \emph{extendable pair of arc-disjoint classic $(r,k)$-cores}. In \autoref{classbranch}, we restate a result from \cite{bhy} that shows that the existence of this substructure is sufficient for the existence of two arc-disjoint $(r,k)$-flow branchings. We then introduce compact $(r,k)$-cores. An \emph{extendable pair of arc-disjoint compact $(r,k)$-cores} corresponds to the substructure $\cal{X}''$. In \autoref{compactbranch}, we show that the existence of an extendable pair of arc-disjoint compact $(r,k)$-cores is also sufficient for the existence of two arc-disjoint $(r,k)$-flow branchings. Again, the proof of \autoref{compactbranch} is our main technical contribution and is similar to the one of \autoref{compactsafe}. Again, having \autoref{compactbranch} at hand, the proof of \autoref{flowfpt} is easy.

Finally, for packing $(r,k)$-safe spanning trees, the substructure ${\cal X}'$ defined in~\cite{bhy} is a \emph{completable pair of edge-disjoint classic $(r,k)$-certificates}. In \autoref{unclass}, we restate a result from \cite{bhy} that shows that the existence of this substructure is sufficient for the existence of two arc-disjoint $(r,k)$-safe spanning trees. We then introduce compact $(r,k)$-certificates. A \emph{completable pair of edge-disjoint compact $(r,k)$-certificates} corresponds to the substructure $\cal{X}''$. In \autoref{compactun}, we show that the existence of a completable pair of edge-disjoint  compact $(r,k)$-certificates is also sufficient for the existence of two edge-disjoint $(r,k)$-safe spanning trees. Again, the proof of \autoref{compactun} is our main technical contribution and is similar to the one of \autoref{compactsafe} and \autoref{compactbranch}. Again, having \autoref{compactun} at hand, the proof of \autoref{treefpt} is easy.



\medskip

\noindent \textbf{Organization}. In \autoref{pre} we review some more technical results we need
from \cite{bhy}, prove several preliminary results, and provide the basic definitions about parameterized complexity. In \autoref{safe}, \autoref{flow},  \autoref{un}, and \autoref{sec:proof-hard-trees} we give the proof of \autoref{safefpt}, \autoref{flowfpt}, \autoref{treefpt}, and
\autoref{theo:lower-graph}, respectively. Finally, we conclude our work in \autoref{con}.

\section{Preliminaries}\label{pre}
In this section we collect some more technical previous results and prove several preliminary statements. We first give some general results on
graphs and digraphs and then some which are more specific to each of
the particular applications. We also provide some basic definitions about parameterized complexity.
\subsection{General preliminaries}

The following is a well-known submodularity property of digraphs that can be found, for instance, in~\cite[Proposition 1.2.1]{book}.
\begin{proposition}\label{subm}
Let $D=(V,A)$ be a digraph and $S_1,S_2 \subseteq V$. Then
$d_D^-(S_1)+d_D^-(S_2)\leq d_D^-(S_1 \cup S_2)+d_D^-(S_1 \cap S_2)$.
\end{proposition}

A rooted digraph $D=(V+r,A)$ is called {\it
  $k$-root-connected} if $d_D^-(X)\geq k$ for all $X \subseteq V$. We
use {\it root-connected} for $1$-root-connected. In a root-connected
rooted digraph $D=(V+r,A)$ an arc $a \in A$ is called {\it critical}
if $D-a$ is not root-connected anymore.  We now use \autoref{subm} to obtain a result that will be
useful when dealing with both $k$-safe spanning $r$-arborescences and spanning $(r,k)$-flow
branchings.
\begin{lemma}\label{utfy}
Let $D=(V+r,A)$ be a $2$-root-connected rooted digraph and let
$D'=(V+r,A')$ be a root-connected rooted digraph that is obtained from
$D$ by deleting $\alpha$ arcs of $A$. Then for any $v \in V+r$, there
are at most $\alpha$ arcs which are critical in $D'$ and whose tail is
$v$.
\end{lemma}
\begin{proof}
We proceed by induction on $\alpha$. The statement is trivial for
$\alpha=0$. We suppose that it holds for all integers up to some
$\alpha$ and show that it also holds for $\alpha + 1$. Let $v \in V+r$
and let $\{a_1,\ldots,a_{\alpha+1}\}$ be a set of arcs in $A$ such
that $D_2=D-\{a_1,\ldots,a_{\alpha+1}\}$ is root-connected. By the
inductive hypothesis, $v$ is the tail of at most $\alpha$ critical
arcs in $D_{1}=D-\{a_1,\ldots,a_{\alpha}\}$. Suppose, for the sake of a
contradiction, that there are two arcs $vw_1,vw_2$ which are critical
in $D_2$, but not in $D_1$. It follows that there are sets
$X_1,X_2\subseteq V$ such that $d^-_{D_2}(X_1), d^-_{D_2}(X_2)=1$,
$d^-_{D_1}(X_1), d^-_{D_1}(X_2)=2$, and $vw_i$ enters $X_i$. Since
$D_2=D_1-a_{\alpha+1}$, it follows that $a_{\alpha+1}$ enters $X_1\cap
X_2$, so $X_1\cap X_2 \neq \emptyset$. As $D_2$ is root-connected, we
have $d_{D_2}^-(X_1\cap X_2)\geq 1$. As $vw_i$ enters $X_i$, we have
$v \in V+r-(X_1 \cup X_2)$ and so both $vw_1$ and $vw_2$ enter $X_1
\cup X_2$. This yields $d^-_{D_2}(X_1)+d^-_{D_2}(X_2)=1+1<1+2\leq
d^-_{D_2}(X_1\cap X_2)+d^-_{D_2}(X_1\cup X_2)$, a contradiction to
\autoref{subm}.
\end{proof}

Given a $2$-root-connected rooted
digraph $D=(V+r,A)$, a pair of subdigraphs
$(X_1=(V_1+r,A_1),X_2=(V_2+r,A_2))$ of $D$ is called {\it extendable}
if both $D-A_1$ and $D-A_2$ are root-connected. The following is an immediate consequence of the fact that checking whether a digraph is root-connected can clearly be done in polynomial time.
\begin{lemma}\label{extend}
Given a rooted digraph $D=(V+r,A)$ and a pair of two subdigraphs
$(X_1=(V_1+r,A_1),X_2=(V_2+r,A_2))$, we can decide in polynomial time
whether $(X_1,X_2)$ is extendable.
\end{lemma}

We now switch to some results in undirected graphs.
\begin{lemma}\label{unsamevertex}
Let $G=(V,E)$ be a graph and let $T$ be a spanning tree of $G$. Let $e=uv \in
E-E(T)$. Then there is some $f \in E(T)$ that is incident to $u$ such
that $T-e+f$ is a spanning tree of $G$.
\end{lemma}
\begin{proof}
The graph $T+f$ contains a unique cycle $C$ such that $uv \in E(C)$
and the deletion of an arbitrary edge of $E(C)$ yields a spanning tree
of $G$. As $C$ is is a cycle, $E(C)$ contains an edge $f$ different
from $e$ that is incident to $u$. This edge satisfies the condition.
\end{proof}
The following result can be found in a stronger form
in~\cite[Theorem 5.3.3]{book}.
\begin{proposition}\label{unsigma}
Let $G=(V,E)$ be a graph and let $T_1,T_2$ be spanning trees of
$G$. Then there is a function $\sigma:E(T_1)\rightarrow E(T_2)$ such
that for all $e \in E(T_1)$ both $T_1-e+\sigma(e)$ and
$T_2-\sigma(e)+e$ are spanning trees of $G$.
\end{proposition}
We call a function like in \autoref{unsigma} a {\it
  tree-mapping function} from $T_1$ to $T_2$.

\begin{lemma}\label{jamais3}
Let $G=(V,E)$ be a graph, let $T_1,T_2$ be spanning trees of $G$, and
let $\sigma:E(T_1)\rightarrow E(T_2)$ be a tree-mapping function from
$T_1$ to $T_2$. Further, let $e_1,e_2,e_3\in E(T_1)$ be all incident
to a common vertex $v$. Then $\{\sigma(e_1),\sigma(e_2),\sigma(e_3)\}$
contains at least two distinct elements.
\end{lemma}
\begin{proof}
As $T_1$ is a spanning tree, $T_1-\{e_1,e_2,e_3\}$ contains three
components $C_1,C_2,C_3$ none of which contains $v$ such  that
$e_i$ is incident to a vertex in $V(C_i)$ for $i=1,2,3$. As
$T_1-e_i+\sigma(e_i)$ is a spanning tree, we obtain that $\sigma(e_i)$
is incident to a vertex in $V(C_i)$ for $i=1,2,3$. As $V(C_1),V(C_2)$
and $V(C_3)$ are pairwise disjoint, the statement follows.
\end{proof}

\subsection{Preliminaries on $k$-safe spanning $r$-arborescences}
\label{prelim:arborescences}

Given a rooted digraph $D=(V+r,A)$, a {\it classic $(r,k)$-kernel} is
a subarborescence $X=(V'+r,A')$ of $D$ such that $X$ is $k$-safe and
$|V'| = 2k-2$. The \XP algorithm of \autoref{safexp} is based on the following result, which we reformulate here using our terminology.
\begin{lemma}[Bang-Jensen, Havet, and Yeo~\cite{bhy}]
\label{classksafe}
Let $D=(V+r,A)$ be a rooted digraph with $|V|\geq 2k-2$. Then $D$
contains two arc-disjoint $k$-safe spanning $r$-arborescences if and
only if $D$ contains an extendable pair of arc-disjoint classic
$(r,k)$-kernels. Further, the two arc-disjoint $k$-safe spanning
$r$-arborescences can be constructed from the extendable pair of
classic $(r,k)$-kernels in polynomial time.
\end{lemma}

\subsection{Preliminaries on spanning $(r,k)$-flow branchings}
\label{prelim:flow}

We first need the following result that allows to recognize
$(r,k)$-flow branchings.

\begin{lemma}[Bang-Jensen, Havet, and Yeo~\cite{bhy}]
\label{classkern}
Given a rooted digraph $D=(V+r,A)$ and a non-negative integer $k$, we can decide in polynomial time
whether $D$ is an $(r,k)$-flow branching.
\end{lemma}

Given a digraph $D=(V,A)$ and two vertices $u,v \in V$, a {\it $uv$-path flow} is a
flow ${\mathbf c}$ such that ${\mathbf c}(a)=1$ for  all arcs $a\in A(P)$ and
${\mathbf c}(a)=0$ for all $a\in A-A(P)$ for some $uv$-path $P$.  Similarly, a
{\it cycle flow} is a flow ${\mathbf c}$ such that ${\mathbf c}(a)=1$  of all $a\in
A(C)$ and ${\mathbf c}(a)=0$ for all $a\in A-A(C)$ for some cycle $C$. We need the following result on flows which is proven in a more
general form in \cite{bg}.

\begin{lemma}\label{decomp}
Let $X=(V+r,A)$ be a flow branching and ${\mathbf z}:A \rightarrow
\mathbb{Z}_{\geq 0}$ be a branching flow in $X$. Then there is an
$rv$-path flow ${\mathbf z}_v$ for all $v \in V$ and a set of cycle flows
$\{{\mathbf z}_C:C \in \mathcal{C}\}$ for some set of cycles $\mathcal{C}$ such
that ${\mathbf z}=\sum_{v\in V}{\mathbf z}_v+\sum_{C \in \mathcal{C}}{\mathbf z}_C$.
\end{lemma}

We now use \autoref{decomp} to prove an important property of
arc-minimal $(r,k)$-flow branchings. A rooted digraph $X=(V+r,A)$ is
called {\it triple-free} if it does not contain more than two arcs in
the same direction between the same two vertices.

\begin{lemma}\label{not3}
Let $X=(V+r,A)$ be an arc-minimal $(r,k)$-flow branching with $|V|\geq
2k-1$. Then $X$ is triple-free.
\end{lemma}

\begin{proof}
Suppose that $X$ contains three arcs $a_1,a_2,a_3$ whose tail is $u$
and whose head is $v$ for some $u,v\in V+r$. Further, let ${\mathbf z}:A
\rightarrow \mathbb{Z}$ be an $(r,k)$-branching flow. By
\autoref{decomp}, we obtain that ${\mathbf z}=\sum_{v\in V}{\mathbf z}_v+\sum_{C \in
  \mathcal{C}}{\mathbf z}_C$ where ${\mathbf z}_v$ is an $rv$-path flow for all $v \in V$
and ${\mathbf z}_C$ is a cycle flow for all $C \in \mathcal{C}$ for some set of cycles $\mathcal{C}$. Let
${\mathbf z}'=\sum_{v\in V}{\mathbf z}_v$. Observe that ${\mathbf z}'$ is an $(r,k)$-branching
flow. As all of the ${\mathbf z}_v$ are path flows, we obtain
${\mathbf z}'(a_1)+{\mathbf z}'(a_2)+{\mathbf z}'(a_3)\leq |V|< 2(n-k)$. It follows that we can
define a flow ${\mathbf z}'':A \rightarrow \mathbb{Z}$ such that
${\mathbf z}''(a_1)+{\mathbf z}''(a_2)={\mathbf z}'(a_1)+{\mathbf z}'(a_2)+{\mathbf z}'(a_3)$, ${\mathbf z}''(a_3)=0$ and
${\mathbf z}''(a)={\mathbf z}'(a)$ for all $a \in A-\{a_1,a_2,a_3\}$. It is easy to see
that ${\mathbf z}''$ is an $(r,k)$-branching flow, so $X-a_3$ is an $(r,k)$-flow
branching, a contradiction to the minimality of $X$.
\end{proof}

Given a rooted digraph $D=(V+r,A)$ with $|V+r|\geq 2k$, a {\it classic
  $(r,k)$-core} is an $(r,k)$-flow branching $X=(V'+r,A')$ that is a
subdigraph of $D$ with $|V'|=2k-1$.  The \XP algorithm of \autoref{flowxp} is based on the following result, again reformulated using our terminology.


\begin{lemma}[Bang-Jensen, Havet, and Yeo~\cite{bhy}]
\label{classbranch}
Let $D=(V+r,A)$ be a rooted digraph with $|V|\geq 2k-1$. Then $D$
contains two arc-disjoint  spanning $(r,k)$-flow branchings if and only if $D$
contains an extendable pair of arc-disjoint classic
$(r,k)$-cores. Further, the two arc-disjoint spanning $(r,k)$-flow branchings
can be constructed in polynomial time from the extendable pair of
arc-disjoint classic $(r,k)$-cores.
\end{lemma}

\subsection{Preliminaries on $(r,k)$-safe spanning trees}
\label{prelim:trees}

We first define an equivalent for extendability in undirected
graphs. Given a rooted graph $G=(V+r,E)$, a pair of subtrees
$(X_1,X_2)$ is called {\it completable}\footnote{Note that we require a completable pair of trees to be edge-disjoint. This is in contrast with the fact that, in \autoref{prelim:arborescences} and \autoref{prelim:flow}, we do {\sl not} require the elements in an extendable pair to be arc-disjoint. We adopt this asymmetric choice for technical reasons arising from the proofs.} if there are edge-disjoint
spanning trees $T_1,T_2$ of $G$ such that $E(X_i)\subseteq
E(T_i)$. The following result that allows to test completabality can
be established using matroid theory as mentioned in \cite{bhy}.
\begin{lemma}[Bang-Jensen, Havet, and Yeo~\cite{bhy}]
\label{uncheck}
Given a graph $G=(V+r,E)$ and a pair of subtrees $(X_1,X_2)$, we can
decide in polynomial time whether $(X_1,X_2)$ is completable.
\end{lemma}

Given a rooted graph $G=(V+r,E)$ with $|V+r|\geq 2k-1$, a {\it classic
  $(r,k)$-certificate} is an $(r,k)$-safe subtree $X=(V'+r,E')$ of $G$
with $|V'|=2k-2$. The \XP algorithm of \autoref{treexp} is based on the following result, again restated using our terminology.


\begin{lemma}[Bang-Jensen, Havet, and Yeo~\cite{bhy}]
\label{unclass}
Let $G=(V+r,E)$ be a rooted graph. Then $G$ contains two edge-disjoint
$(r,k)$-safe spanning trees if and only if $G$ contains a completable
pair of classic $(r,k)$-certificates. Further, given a completable
pair of classic $(r,k)$-certificates, we can compute two edge-disjoint
$(r,k)$-safe spanning trees in polynomial time.
\end{lemma}


\subsection{Preliminaries on parameterized complexity}
\label{prelim:pc}


 We refer the reader to~\cite{FPT-book} for basic background on parameterized complexity, and we recall here only some basic definitions used in this article.
A \emph{parameterized problem} is a decision problem whose instances are pairs $(x,k) \in \Sigma^* \times \mathbb{N}$, where $k$ is called the \emph{parameter}.
A parameterized problem $L$ is \emph{fixed-parameter tractable} ({\sf FPT}) if there exists an algorithm $\Acal$, a computable function $f$, and a constant $c$ such that given an instance $I=(x,k)$,
$\Acal$ (called an {\sf FPT} \emph{algorithm}) correctly decides whether $I \in L$ in time bounded by $f(k) \cdot |I|^c$. A parameterized problem $L$ is \emph{slice-wise polynomial} ({\sf XP}) if there exists an algorithm $\Acal$ and two computable functions $f,g$ such that given an instance $I=(x,k)$,
$\Acal$ (called an {\sf XP} \emph{algorithm}) correctly decides whether $I$ is a positive instance of $L$ in time bounded by $f(k) \cdot |I|^{g(k)}$.

A \emph{kernelization algorithm}, or just \emph{kernel}\footnote{This notion should {\sl not} be confused with (classic or compact) $(r,k)$-kernels used throughout this article.},  for a parameterized problem $L$ is an algorithm $\mathcal{A}$ that, in polynomial time, generates from an instance $I=(x, k)$ of $L$ an equivalent instance $I'=(x', k')$ of $L$ such that $|x'| + k' \leq f(k)$, for some computable function $f : \mathbb{N} \rightarrow \mathbb{N}$. If $f(k)$ is bounded from above by a polynomial of the parameter $k$, we say that $L$ admits a \emph{polynomial kernel}. It is well-known~\cite{FPT-book} that a parameterized problem is \FPT if and only if it admits a (not necessarily polynomial) kernel.

\section{An FPT algorithm for packing $k$-safe spanning arborescences}\label{safe}

This section is concerned with proving \autoref{safefpt}. Given a rooted digraph $D=(V+r,A)$ and a positive integer $k$, we say that a vertex $v\in V+r$ is
{\it large} if $|N_D^+(v)|\geq 6k-5$, and {\it small} otherwise. We let
{\boldmath $L_D$} (resp. $S_D$) be the set of vertices in $V+r$ which are large
(resp. small) in $D$.

We are now ready to introduce a new notion of $(r,k)$-kernels for $k$-safe spanning
$r$-arborescences.  A {\it compact $(r,k)$-kernel} is a subdigraph
$X=(V'+r, A')$ of $D$ with $|V'|\leq 2k-2$ satisfying the following:
\begin{itemize}
\item $X$ is an $r$-arborescence,
\item all vertices in $V'\cap L_D$ are sinks of $X$, and
\item a $k$-safe $r$-arborescence $Y$ can be obtained from $X$ by
  adding a set $V^*$ of $2k-2-|V'|$ new vertices and adding an arc
  from a vertex in $V'\cap L_D$ to $v$ for all $v \in V^*$.
\end{itemize}


Observe that $Y$ is not necessarily a subdigraph of $D$. The following
result shows that compact $(r,k)$-kernels can be used in a similar way as
classic $(r,k)$-kernels.

\begin{lemma}\label{compactsafe}
Let $D=(V+r,A)$ be a $2$-root-connected rooted digraph with $|V|\geq
2k-2$. Then $D$ contains two arc-disjoint $k$-safe spanning
$r$-arborescences if and only if $D$ contains an extendable pair of
arc-disjoint compact $(r,k)$-kernels. Further, given an extendable pair
of arc-disjoint compact $(r,k)$-kernels, we can find a pair of
arc-disjoint $k$-safe spanning $r$-arborescences in
polynomial time.
\end{lemma}
\begin{proof}

By \autoref{classksafe}, for the first part it suffices to prove
that $D$ contains a pair of arc-disjoint extendable compact
$(r,k)$-kernels if and only if $D$ contains a pair of arc-disjoint
extendable classic $(r,k)$-kernels.

First let $(X_1=(V_1+r,A_1),X_2=(V_2+r,A_2))$ be an extendable pair of
arc-disjoint classic $(r,k)$-kernels. Let $X_i'=(V_i'+r,A_i')$ be
obtained from $X_i$ by deleting $B^v_{X_i}-v$ for all $v \in
(V_i+r)\cap L_D$. By construction, the $X_i'$ are $r$-arborescences
and all vertices in $(V_i'+r)\cap L_D$ are sinks in $X_i'$.
 Let $Y_i$
be obtained from $X'_i$ by adding the vertices in $V_i-V'_i$ and
adding an arc from a vertex $v \in (V'_i+r)\cap L_D$ to a vertex $w
\in V_i-V'_i$ whenever $w \in V(B_{X_i}^v)-v$. Observe that $Y_i$ is
an arborescence with $|V(Y_i)|=2k-1$. Further, note that
$|V_i+r-V(B_{Y_i}^v)|= |V_i+r-V(B_{X_i}^v)|\geq k$ for all $v \in
V'_i$ and $|V(B_{Y_i}^v)|= 0$ for all $v \in V_i-V_i'$. This yields
that $|V_i+r-V(B_{Y_i}^v)|\geq|V_i+r-V(B_{X_i}^v)|\geq k$ for all $v
\in V_i$ and so $Y_i$ is a $k$-safe arborescence. By definition, we
obtain that $(X_1',X_2')$ is a pair of arc-disjoint compact
$(r,k)$-kernels. Further, $D-A_i$ is a subdigraph of $D-A'_i$ that is
root-connected, so $D-A_i'$ is root-connected as well. This yields that
$(X_1',X_2')$ is extendable.

Now let $(X_1=(V_1+r,A_1),X_2=(V_2+r,A_2))$ be an extendable pair of
arc-disjoint compact $(r,k)$-kernels. By definition, there are $k$-safe
arborescences $Y_1,Y_2$ such that $Y_i$ is obtained from $X_i$ by
adding a set $V_i^*$ of $2k-2-|V_i|$ new vertices and an arc from a
vertex in $V_i\cap L_D$ to $v$ for all $v \in V_i^*$. Let
$(X_1'=(V_1'+r,A_1'),X_2'=(V_2'+r,A_2'))$ be a pair of subdigraphs of
$D$ that are vertex-maximal with the following properties:
\begin{enumerate}[(i)]
\item[\textbf{(i)}] $X_i'$ is obtained from $X_i$ by repeatedly adding another vertex $v \in V-V_i$ and an arc of $A$ that goes from a vertex in $V_i\cap L_D$ to $v$,
\item[\textbf{(ii)}] $d_{X'_i}^+(v)\leq d_{Y_i}^+(v)$ for all $v \in V_i\cap L_D$,
\item[\textbf{(iii)}] $A'_1$ and $A'_2$ are disjoint, and
\item[\textbf{(iv)}] $(X_1',X_2')$ is extendable.
\end{enumerate}

Note that $(X_1,X_2)$ satisfies  conditions \textbf{(i)}-\textbf{(iv)}, so
$(X_1',X_2')$ is well-defined. Also observe that if condition
\textbf{(ii)} is satisfied with equality for all $v \in V_i\cap L_D$, then
$X'_i$ is isomorphic to $Y_i$, so $X'_i$ is a $k$-safe arborescence
and by definition also a classic $(r,k)$-kernel. If this is the case
for both $X'_1$ and $X'_2$, we are done by conditions\textbf{ (iii)} and
\textbf{(iv)}.

We may therefore suppose by symmetry that there is a vertex $v \in
V_1\cap L_D$ with $d_{X'_1}^+(v)< d_{Y_1}^+(v)$. For any $a=vz_a$ with
$z_a \in N_D^+(v)-V_1'$, let $X_a=(V_1'\cup z_a,A_1'\cup a)$. By the
maximality of $X_1'$, we obtain that $X_a$ violates one of
conditions \textbf{(i)}-\textbf{(iv)} for all $a \in \delta_D^+(v)$ with $z_a \in
N_D^+(v)-V_1'$. By construction and the choice of $v$, $X_a$ satisfies
conditions \textbf{(i)} and \textbf{(ii)} for all $a \in \delta_D^+(v)$ with
$z_a \in N_D^+(v)-V_1'$. If $X_a$ violates \textbf{(iii)}, then $z_a \in
V'_{2}$. By \textbf{(ii)}, we have $|V_{2}'|\leq |V_2 \cup V_2^*|=2k-2$ and
so this is the case for at most $2k-2$ vertices $z_a \in
N_D^+(v)-V_1'$. If $X_a$ does not satisfy \textbf{(iv)}, then $a$ is critical
in $D-A'_1$. As $D$ is 2-root-connected and $X'_1$ is an arborescence
by construction, \autoref{utfy} implies that this is the case for at
most $|A'_1|=|V_1'|\leq |V_1\cup V_1^*|=2k-2$ vertices $z_a \in
N_D^+(v)-V_1'$. As $v\in L_D$ and $|V_1'|\leq |V_1\cup V_1^*|=2k-2$,
we have $|N_D^+(v)-V_1'|\geq (6k-5)-(2k-2)>(2k-2)+(2k-2)$, so there is
at least one vertex in $z \in N_D^+(v)-V_1'$ and an arc $a=vz$ such
that $X_a$ does not violate any of conditions \textbf{(i)}-\textbf{(iv)}, a
contradiction.

Observe that the second part of the proof yields an algorithm for
computing a pair of arc-disjoint extendable classic
$(r,k)$-kernels. Indeed, every time we try to add an arc $a=vz$ to $X'_i$, we
test if $(X_a,X'_{3-i})$ satisfies  conditions \textbf{(i)}-\textbf{(iv)}. Conditions \textbf{(i)}-\textbf{(iii)} can clearly be checked in polynomial time and, by
\autoref{extend}, condition \textbf{(iv)} can also be checked in
polynomial time. Never testing an arc that is parallel to one that we have
tested already, after at most $4k-4$ failed attempts, we manage to add
a new vertex to $V_i'$. We repeat this procedure $|V_i^*|\leq 2k-2$
times. It follows that a pair of arc-disjoint extendable classic
$(r,k)$-kernels can be computed in time $k^2\cdot n^{O(1)} =n^{O(1)}$. By the second
part of \autoref{classksafe}, we can then find the arc-disjoint
$k$-safe spanning arborescences in $D$ in time $n^{O(1)}$. Therefore, the overall running time
of the algorithm is polynomial, as claimed.
\end{proof}
We are now ready to proceed to the proof of \autoref{safefpt}.


\begin{proof1}

We may suppose that there are at most two parallel arcs from $u$
 to $v$ for any $u,v \in V$. If $|V|<2k-2$, the problem can be solved by a brute force algorithm in time $2^{O(k^2)}$, by generating all pairs of  subdigraphs of $D$ and checking whether any of these pairs satisfies the required conditions.
We may hence also suppose that
$|V|\geq 2k-2$.

We can first decide in time $n^{O(1)}$ if $D$
 is $2$-root-connected. If it is not, the answer is negative, so we
 may suppose it is.  Let $X_1=(V_1+r,A_1),X_2=(V_2+r,A_2)$ be two
 subdigraphs of $D$ with $|V_1|,|V_2|\leq 2k-2$. In order to test whether $X_i$ is a compact
 $(r,k)$-kernel, we first verify if $X_i$ is an $r$-arborescence such
 that all the vertices in $V_i \cap L_D$ are sinks in $X_i$. If this
 is the case, we add a set $V_i^*$ of $2k-2-|V_i|$ new vertices to
 $X_i$. We then test all possibilities to add one arc from
 $(V_i+r)\cap L_D$ to $v$ for all $v \in V_i^*$. As $|V_i+r|\leq 2k$
 and $|V_i^*|\leq 2k$, there are at most $2k^{2k} = 2^{O(k \cdot \log k)}$ possibilities to
 check. For each of these possibilities, we can then check in polynomial time whether the obtained graph is a $k$-safe arborescence.

 For each such pair $X_1,X_2$, by the definition of classic $(r,k)$-kernels, we can therefore check
 in time $2^{O(k \cdot \log k)}$ whether both $X_1$ and $X_2$
 are compact $(r,k)$-kernels. By \autoref{extend}, we can therefore
 decide in  time $2^{O(k \cdot \log k)}+n^{O(1)}$
 if
 $(X_1,X_2)$ is an extendable pair of arc-disjoint compact
 $(r,k)$-kernels in $D$. By \autoref{compactsafe}, it therefore
 suffices to prove that there are at most  $2^{O(k^2 \cdot \log k)}$ possible candidates
 for the extendable pair of arc-disjoint compact $(r,k)$-kernels,
and that these can be generated within this running time.

Let $X=(V'+r,A')$ be a compact $(r,k)$-kernel in $D$. Observe that every
vertex in $V'$ can be reached from $r$ by a directed path all of whose
interior vertices are in $S_D$ and whose length is at most $k-1$.
As every vertex in $S_D$ has at most $6k-6$
out-neighbors, we obtain that the number of vertices that can be
reached by such a path is at most
$1+(6k-6)+(6k-6)^2+\ldots+(6k-6)^{k-1}\leq (6k)^{k}$. As $V'$ contains
at most $2k-2$ vertices, there are at most ${(6k)^{k} \choose
  2k-2}\leq (6k)^{2k^2}$ possibilities to choose~$V'$.

 Now suppose that we have chosen $V'$ of size $2k-2$. As there are at
 most two arcs in the same direction between any two vertices, there
 are at most $4{|V'+r| \choose 2}\leq 16 k^2$ arcs that have their
 head and tail in $V'+r$. As $|A'|= 2k-2$, there are at most ${16 k^2
   \choose 2k-2}\leq (16 k^2)^{2k} $ possibilities to choose $A'$. It
 follows that there are at most $(6k)^{2k^2}\cdot (16 k^2)^{2k}$
 possibilities to choose a compact $(r,k)$-kernel $X$. As these can be
 computed by a brute force method, the algorithm can finish after
 checking less than $f(k)={(6k)^{2k^2}\cdot (16 k^2)^{2k} \choose 2} = 2^{O(k^2 \cdot \log k)}$
 candidates for the extendable pair of arc-disjoint compact
 $(r,k)$-kernels.

 If no extendable pair of arc-disjoint compact $(r,k)$-kernels exists,
 by \autoref{compactsafe}, $D$ does not contain two arc-disjoint
 $k$-safe spanning arborescences. On the other hand, once we have found an extendable pair of arc-disjoint compact $(r,k)$-kernels, we can compute the two arc-disjoint $k$-safe spanning $r$-arborescences in polynomial time by the second part of \autoref{compactsafe}. The overall running time of the obtained algorithm is $2^{O(k^2 \cdot \log k)} \cdot n^{O(1)}$.
\end{proof1}

\section{An FPT algorithm for packing  spanning $(r,k)$-flow branchings}\label{flow}

This section is concerned with proving \autoref{flowfpt}.
Slightly modifying the terminology introduced in \autoref{safe}, given a 2-root-connected digraph $D=(V+r,A)$ and a positive integer $k$, we say that a vertex $v$
is {\it large} if $|N^+(v)|\geq 20k^2+1$, and {\it small} otherwise. Again, we
let {\boldmath $L_D$} (resp. $S_D$) be the set of vertices in $V+r$ which are
large (resp. small) in $D$.

We are now ready to introduce a new notion of $(r,k)$-cores for spanning $(r,k)$-flow
branchings.  A {\it compact $(r,k)$-core} is a subdigraph $X=(V'+r, A')$
of $D$ with $|V'|\leq 2k-1$ satisfying the following:
\begin{itemize}
\item all vertices in $V'\cap L_D$ are sinks in $X$ and
\item an $(r,k)$-flow branching $Y$ can be obtained from $X$ by adding
  a set $V^*$ of $2k-1-|V'|$ new vertices and adding an arc from a
  vertex in $V'\cap L_D$ to $v$ for all $v \in V^*$.
\end{itemize}

Observe that $Y$ is not necessarily a subdigraph of $D$. The following
result, which is similar to \autoref{compactsafe}, shows that compact $(r,k)$-cores can be used in a similar way as classic $(r,k)$-cores.

\begin{lemma}\label{compactbranch}
Let $D=(V+r,A)$ be a $2$-root-connected rooted digraph with $|V|\geq
2k-1$. Then $D$ has two arc-disjoint spanning $(r,k)$-flow branchings if and
only if $D$ contains an extendable pair of arc-disjoint compact
$(r,k)$-cores. Further, given an extendable pair of triple-free
arc-disjoint compact $(r,k)$-cores, we can find a pair of arc-disjoint
spanning $(r,k)$-flow branchings in polynomial time.
\end{lemma}
\begin{proof}

By \autoref{classbranch}, for the first part it suffices to prove
that $D$ contains an extendable pair of arc-disjoint compact
$(r,k)$-cores if and only if $D$ contains an extendable pair of
arc-disjoint classic $(r,k)$-cores.

 First let $(X_1=(V_1+r,A_1),X_2=(V_2+r,A_2))$ be an extendable pair
 of arc-disjoint classic $(r,k)$-cores. Let $X_i'=(V_i'+r,A'_i)$ be
 obtained from $X_i$ by first deleting all arcs in $A_i$ whose tail is
 a large vertex and then restricting to the subdigraph that is
 root-connected from $r$. As $A'_i\subseteq A_i$, $(X_1',X_2')$ is
 extendable. It remains to show that $X_1'$ and $X_2'$ are compact
 $(r,k)$-cores. By construction, all vertices in $V'_i\cap L_D$ are
 sinks in $X'_i$. As $X_i$ is an $(r,k)$-flow branching, there is
 an $(r,k)$-branching flow ${\mathbf z}:A_i\rightarrow \mathbb{Z}$ in
 $X_i$. Create $Y_i$ from $X_i$ by attaching ${\mathbf z}(\delta^+_{X_i}(v))$
 arcs
  directed away from $v$ to every vertex $v \in
 V'_i\cap L_D$. Assigning ${\mathbf z}'(a)=1$ for all arcs leaving a large vertex
 in $Y_i$ and ${\mathbf z}'(a)={\mathbf z}(a)$ for all remaining arcs, we obtain that ${\mathbf z}'$
 is an $(r,k)$-branching flow in $Y_i$, so $Y_i$ is an $(r,k)$-flow
 branching. Furthermore, we have
 $|V(Y_i)-r|={\mathbf z}'(\delta^+_{Y_i}(r))={\mathbf z}(\delta^+_{X_i}(r))=|V_i|=2k-1$. It
 follows by definition that $X_i'$ is a compact $(r,k)$-core.

Now let $(X_1=(V_1+r,A_1),X_2=(V_2+r,A_2))$ be an extendable pair of
arc-disjoint compact $(r,k)$-cores. Possibly deleting arcs, we may
suppose by \autoref{not3} and as $|V|\geq 2k-1$ that $X_1$ and $X_2$
are triple-free. By definition, there are $(r,k)$-flow branchings
$Y_1,Y_2$ such that $Y_i$ is obtained from $X_i$ by adding a set
$V_i^*$ of $2k-1-|V_i|$ new vertices and adding an arc from a vertex
in $V_i\cap L_D$ to $v$ for all $v \in V_i^*$.  Let
$X_1'=(V_1'+r,A_1')$ and $X_2'=(V_2'+r,A_2')$ be subdigraphs of $D$
that are vertex-maximal with the following properties:
\begin{enumerate}[(i)]
\item[\textbf{(i)}] $X_i'$ is obtained from $X_i$ by repeatedly adding another
  vertex in $v \in V-V_i$ and an arc of $A$ that goes from a vertex in
  $V_i\cap L_D$ to $v$,
\item[\textbf{(ii)}] $d_{X'_i}^+(v)\leq d_{Y_i}^+(v)$ for all $v \in V_i\cap L_D$,
\item[\textbf{(iii)}] $A'_1$ and $A'_2$ are disjoint, and
\item[\textbf{(iv)}] $(X_1',X_2')$ is extendable.
\end{enumerate}

Note that $(X_1,X_2)$ satisfies conditions \textbf{(i)}-\textbf{(iv)}, so $(X_1',X_2')$ is well-defined. 
Further, observe that if condition \textbf{(ii)} is satisfied with equality for all $v \in V_i\cap L_D$, then $X'_i$ is isomorphic to $Y_i$, so $X'_i$ is an $(r,k)$-flow branching, thus by definition also a classic $(r,k)$-core. If this is the case for both $X'_1$ and $X'_2$, we are done by conditions \textbf{(iii)} and \textbf{(iv)}.

We may therefore suppose by symmetry that there is some $v \in V_1\cap
L_D$ with $d_{X'_1}^+(v)< d_{Y_1}^+(v)$. For any $a=vz_a$ with $z_a
\in N_D^+(v)-V_1'$, let $X_a=(V_1'\cup z_a,A_1'\cup a)$. By the
maximality of $X'_1$, $X_a$ violates one of conditions \textbf{(i)}-\textbf{(iv)}
for all $a \in \delta_D^+(v)$ with $z_a \in N_D^+(v)-V'_1$. By
construction and the choice of $v$, $X_a$ satisfies conditions
\textbf{(i)} and \textbf{(ii)} for all $a \in \delta_D^+(v)$ with $z_a \in
N_D^+(v)-V'_1$. If $X_a$ violates condition \textbf{(iii)}, then $z_a \in
V'_{2}$. As $|V_{2}'|\leq |V(Y_{2})|=2k$, this is the case for at most
$2k$ vertices in $N_D^+(v)-V'_1$. If $X_a$ does not satisfy \textbf{(iv)},
then $a$ is critical in $D-A'_1$. As $X_1$ is triple-free and $|V_i'|\leq
2k$ and by construction, we obtain that $|A'_1|\leq 4{|V'_1|\choose
  2}\leq 16 k^2$. Now \autoref{utfy} implies that this is the case
for at most $16 k^2$ vertices in $N_D^+(v)-V'_1$. As $v\in L_D$ and
$|V_1'|\leq |V_1\cup V_1^*|=2k-1$, we have $|N_D^+(v)-V_1'|\geq
(20k^2+1)-(2k-1)>2k+16k^2$, so there is at least one vertex $z \in
N_D^+(v)-V_1'$ and an arc $a=vz$ such that $X_a$ does not violate any
of conditions \textbf{(i)}-\textbf{(iv)}, a contradiction.

Observe that the second part of the proof yields an algorithm for
computing a pair of arc-disjoint extendable classic
$(r,k)$-cores. Indeed, every time we try to add an arc $a=vz$ to $X'_i$, we test
if $(X_a,X'_{3-i})$ satisfies conditions \textbf{(i)}-\textbf{(iv)}. Conditions \textbf{(i)}-\textbf{(iii)} can clearly be checked in polynomial time and, by
\autoref{extend}, condition \textbf{(iv)} can also be checked in
polynomial time. Never checking an arc $a$ which is parallel to an arc we
have already checked, after at most $20k^2$ failed attempts, we manage
to add a new vertex to $V_i'$. We repeat this procedure at most
$|V^*_i|\leq 2k$ times. It follows that a pair of arc-disjoint
extendable classic $(r,k)$-cores can be computed in time $40k^3\cdot
n^{O(1)} = n^{O(1)}$. By the second part of \autoref{classbranch}, we can then
find the arc-disjoint spanning $(r,k)$-flow branchings in $D$ in polynomial time. The
overall running time of the algorithm is polynomial, as claimed.
\end{proof}
We are now ready to proceed to the proof of \autoref{flowfpt}.

\begin{proof2}
First consider the case that $|V|<2k-1.$ Observe that any arc-minimal spanning $(r,k)$-flow branching has at most $|V|$ parallel arcs between any two vertices. It follows that, for any two vertices $u,v$, at most $\gamma:=4k^2$ different distributions of the arcs between $u$ and $v$ among the two candidates for the spanning $(r,k)$-flow branchings have to be considered, including taking none of these arcs. Since there are $\mu:={|V| \choose 2} = O(k^2)$ pairs of vertices, the total number of choices for these distributions is $\gamma^\mu = 2^{O(k^2 \cdot \log k)}$. The problem can therefore be solved by a brute force algorithm  in time $2^{O(k^2 \cdot \log k)}$, by generating all $ 2^{O(k^2 \cdot \log k)}$ pairs of candidate subdigraphs of $D$ and checking whether any of these pairs satisfies the required conditions.
We may hence suppose that $|V|\geq 2k-1$.

By \autoref{not3}, we may also suppose that there are at most four
parallel arcs between any two vertices in $D$. We can first decide in
polynomial time if $D$ is $2$-root-connected. If it is not, the answer is
negative, so we may suppose it is.  Let
$X_1=(V_1+r,A_1),X_2=(V_2+r,A_2)$ be two subdigraphs of $D$ with $|V_1|,|V_2| \leq 2k-1$. In order
to test whether $X_i$ is a compact $(r,k)$-core, we first test if all vertices in $V_i \cap L_D$ are sinks in
$X_i$. We then add a set $V_i^*$ of $2k-1-|V_i|$ new vertices to
$X_i$. We then test all possibilities to add one arc from $(V_i+r)\cap
L_D$ to $v$ for all $v \in V_i^*$. As $|V_i|,|V_i^*|\leq 2k$, there
are at most $2k^{2k} = 2^{O(k \cdot \log k)}$ possibilities to check.
By \autoref{classkern}, we can check in time polynomial in $k$ whether each of the resulting graphs is an $(r,k)$-flow branching. Thus, we can
check in time  $2^{O(k \cdot \log k)}$ whether
both $X_1$ and $X_2$ are compact $(r,k)$-cores.  By \autoref{extend},
we can therefore decide in time $2^{O(k \cdot \log k)} + n^{O(1)}$  if $(X_1,X_2)$ is an extendable pair of arc-disjoint compact
$(r,k)$-cores in $D$. By  \autoref{compactbranch}, it therefore
suffices to prove that there are at most $2^{O(k^2 \cdot \log k)}$ possible candidates
for the extendable pair of arc-disjoint compact $(r,k)$-cores, and that these can be generated within the same running time.

Let $X=(V'+r,A')$ be a compact $(r,k)$-core in $D$. Observe that every
vertex in $V'$ can be reached from $r$ by a directed path all of whose
internal vertices are in $S_D$ and whose length is at most $2k-1$.
As every vertex in $S_D$ has at most $20k^2$
out-neighbors, we obtain that the number of vertices that can be reached
by such a path is at most $1+20k^2+(20k^2)^2+\ldots+(20k^2)^{2k-1}\leq
(20k^2)^{2k}$. As $V'$ contains at most $2k-1$ vertices, there are at
most ${(20k^2)^{2k} \choose (2k-1)}\leq (20k^2)^{4k^2}$ possibilities
to choose $V'$.
Now suppose that we have chosen $V'$ of size at most $2k-1$. As there
are at most four arcs in the same direction between any two vertices,
there are at most $8{|V'| \choose 2}\leq 32 k^2$ arcs that have their
head and tail in $V'+r$. As all arcs of $A'$ have both ends in $V'+r$,
there are at most $2^{32k^2} $ possibilities to choose $A'$. It
follows that there are at most $(20k^2)^{4k^2}\cdot 2^{32k^2}$
possibilities to choose a compact $(r,k)$-core $X$. As these can be
computed by a brute force method, the algorithm can finish after
checking less than $f(k)={(20k^2)^{4k^2}\cdot 2^{32k^2} \choose 2} = 2^{O(k^2 \cdot \log k)}$
candidates for the extendable pair of compact $(r,k)$-cores.

 If no such extendable pair of arc-disjoint compact $(r,k)$-cores exists,
 by \autoref{compactbranch}, $D$ does not contain two arc-disjoint
 spanning $(r,k)$-flow branchings. On the other hand, if we find an extendable pair of arc-disjoint compact
extendable $(r,k)$-cores, we also find such a pair $(X_1,X_2)$ where
$X_1$ and $X_2$ are arc-minimal, so by \autoref{not3}
triple-free. By the second part of \autoref{compactbranch}, we can
compute the two arc-disjoint spanning $(r,k)$-flow branchings in polynomial time. The overall running time of the obtained algorithm is $2^{O(k^2 \cdot \log k)} \cdot n^{O(1)}$.
\end{proof2}

\section{An FPT algorithm for packing $(r,k)$-safe spanning trees}\label{un}

This section is concerned with proving \autoref{treefpt}. Again, slightly modifying the terminology introduced in \autoref{safe} and reused in \autoref{flow}, given a rooted graph $G=(V+r,E)$ and a positive integer $k$, we say that a vertex $v\in V+r$ is
{\it large} if $|N_G(v)|\geq 8k-7$, and {\it small} otherwise. And again, we let
{\boldmath $L_G$} (resp. $(S_G)$) be the set of vertices in $V$ which are large
(resp. small) in $G$.

We are now ready to introduce a new notion of certificates for
$(r,k)$-safe spanning trees. A {\it compact $(r,k)$-certificate} is a
subgraph $X=(V'+r, E')$ of $G$ with $|V'|\leq 2k-2$ satisfying the
following:
\begin{itemize}
\item $X$ is a tree,
\item all vertices in $V'\cap L_G$ are leaves of $X$, and
\item an $(r,k)$-safe spanning tree $Y$ can be obtained from $X$ by
  adding a set $V^*$ of $2k-2-|V'|$ new vertices and adding an edge
  from a vertex in $V'\cap L_G$ to $v$ for all $v \in V^*$.
\end{itemize}

Observe that $Y$ is not necessarily a subgraph of $G$.  The following
result, which is similar to \autoref{compactsafe} and \autoref{compactbranch}, shows that compact certificates can be used in a similar way as
classic certificates.

\begin{lemma}\label{compactun}
Let $G=(V+r,E)$ be a rooted graph with $|V|\geq 2k-2$. Then $G$ has
two edge-disjoint $(r,k)$-safe spanning trees if and only if $G$
contains a completable pair of compact $(r,k)$-certificates. Further,
given a completable pair of compact $(r,k)$-certificates, we can find
a pair of edge-disjoint $(r,k)$-safe spanning trees in polynomial time.
\end{lemma}
\begin{proof}

By \autoref{unclass}, for the first part it suffices to prove that
$G$ contains a completable pair of compact $(r,k)$-certificates if and
only if $G$ contains a pair of completable classic
$(r,k)$-certificates.

First let $(X_1=(V_1+r,E_1),X_2=(V_2+r,E_2))$ be a completable pair of
classic $(r,k)$-certificates. Let $X_i'=(V_i'+r,E_i')$ be obtained
from $X_i$ by deleting $C^v_{X_i}$ for all $v \in (V_i+r)\cap L_G$. By
construction, the $X_i'$ are trees and all vertices in $(V_i'+r)\cap
L_G$ are leaves in $X_i'$. Let $Y_i$ be obtained from $X'_i$ by adding
the vertices in $V_i-V'_i$ and adding an edge from a vertex $v \in
(V'_i+r)\cap L_G$ to a vertex $w \in V_i-V'_i$ whenever $w \in
V(C_{X_i}^v)$. Observe that $Y_i$ is a tree with
$|V(Y_i)|=2k-1$. Further, we have $|V_i-V(C_{Y_i}^v)|=
|V_i-V(C_{X_i}^v)|\geq k$ for all $v \in V'_i$ and $|V(C_{Y_i}^v)|= 0$
for all $v \in V_i-V_i'$. This yields that
$|V_i-V(C_{Y_i}^v)|\geq|V_i-V(C_{X_i}^v)|\geq k$ for all $v \in V_i$
and so $Y_i$ is an $(r,k)$-safe tree. By definition, we obtain that
$(X_1',X_2')$ is a pair of compact $(r,k)$-certificates. Further, as
$E(X'_i)\subseteq E(X_i)$ and $(X_1,X_2)$ is completable, we obtain
that $(X_1',X_2')$ is completable.

Now let $(X_1=(V_1+r,E_1),X_2=(V_2+r,E_2))$ be a completable pair of
compact $(r,k)$-certificates. By definition, there are $(r,k)$-safe
trees $Y_1,Y_2$ such that $Y_i$ is obtained from $X_i$ by adding a set
$V_i^*$ of $2k-2-|V_i|$ new vertices and an edge from a vertex in
$V_i\cap L_G$ to $v$ for all $v \in V_i^*$. Let
$(X_1'=(V_1'+r,E_1'),X_2'=(V_2'+r,E_2'))$ be a pair of subgraphs of
$G$ that are vertex-maximal with the following properties:
\begin{enumerate}[(i)]
\item[\textbf{(i)}] $X_i'$ is obtained from $X_i$ by repeatedly adding another
  vertex $v \in V-V_i$ and an edge of $E$ that goes from a vertex in
  $V_i\cap L_G$ to $v$,
\item[\textbf{(ii)}] $d_{X'_i}(v)\leq d_{Y_i}(v)$ for all $v \in V_i\cap L_G$, and
\item[\textbf{(iii)}] $(X_1',X_2')$ is completable.
\end{enumerate}

Note that $(X_1,X_2)$ satisfies conditions \textbf{(i)}-\textbf{(iii)}, so
$(X_1',X_2')$ is well-defined. Observe that if condition \textbf{(ii)} is
satisfied with equality for all $v \in V_i\cap L_G$, then $X'_i$ is
isomorphic to $Y_i$, so $X'_i$ is an $(r,k)$-safe spanning tree and by
definition also a classic $(r,k)$-certificate. If this is the case for
both $X'_1$ and $X'_2$, we are done by condition \textbf{(iii)}.

We may therefore suppose by symmetry that there is a vertex $v \in
V_1\cap L_G$ with $d_{X'_1}(v)< d_{Y_1}(v)$. For any $e=vz_e \in
\delta_G(v)$ with $z_e \in N_G(v)-V_1'$, let $X_e=(V_1'\cup z_e,E_1'\cup
e)$. By the maximality of $X_1'$, we obtain that $X_e$ violates one of
conditions \textbf{(i)}-\textbf{(iii)} for all $e=vz_e \in \delta_G(v)$ with $z_e
\in N_G(v)-V_1'$. By construction and the choice of $v$, $X_e$
satisfies conditions \textbf{(i)} and \textbf{(ii)} for all $e=vz_e \in
\delta_G(v)$ with $z_e \in N_G(v)-V_1'$. It follows that $(X_e,X'_2)$
violates condition \textbf{(iii)} for all $e=vz_e \in \delta_G(v)$ with
$z_e \in N_G(v)-V_1'$. As $(X_1',X_2')$ is completable, there are two
disjoint spanning trees $T_1,T_2$ of $G$ such that $E_i'\subseteq
E(T_i)$ for $i=1,2$.
\begin{Claim}\label{few1}
There is no $e=vz_e \in \delta_G(v)-(E(T_1) \cup E(T_2))$ with $z_e \in
N_G(v)-V_1'$ such that $(X_e,X'_2)$ violates condition \textbf{\emph{(iii)}}.
\end{Claim}
\begin{proof}
Suppose otherwise. By \autoref{unsamevertex}, there is an edge $f\in
E(T_1)$ incident to $z_e$ such that $T_1'=T_1-f+e$ is a spanning tree
of $G$. As $z_e \notin V_1'$, we obtain that $f \notin E_1'$, yielding
$E(X_e) \subseteq E(T_1')$. As $T_1'$ and $T_2$ are edge-disjoint, we
obtain that $(X_z,X_2')$ is completable, a contradiction.
\end{proof}
\begin{Claim}\label{few2}
There are at most $6k-6$ vertices $z \in N_G(v)-V_1'$ such that
$(X_e,X'_2)$ violates condition \textbf{\emph{(iii)}} for some $e=vz \in
E(T_2)$.
\end{Claim}
\begin{proof}
Suppose otherwise. As $|V_2'|\leq |V(Y_2)|-1 =2k-2$, we obtain that
there are at least $4k-3$ vertices $z \in N_G(v)-V_1'$ such that
$(X_e,X'_2)$ violates condition \textbf{(iii)} for some $e=vz \in
E(T_2)-E_2$. Let $\sigma:E(T_2)\rightarrow E(T_1)$ be a tree-mapping
function from $T_2$ to $T_1$. By \autoref{jamais3} and since $|E'_1|\leq
2k-2$, there is some $z \in N_G(v)-V_1'$ and an edge $e=vz \in
E(T_2)-E_2$ such that $\sigma(e)\in E(T_1)-E_1$. By definition of
tree-mapping functions, $T_1'=T_1-\sigma(e)+e$ and
$T_2'=T_2-e+\sigma(e)$ are edge-disjoint spanning trees of $G$. As
$E(X_e)\subseteq E(T_1')$ and $E(X_2)\subseteq E(T_2')$, we obtain
that $(X_e,X_2')$ is completable, a contradiction.
\end{proof}
 As $v\in L_G$ and $|V_1'|\leq |V_1\cup V_1^*|=2k-2$, we have
 $|N_G(v)-V_1'|\geq (8k-7)-(2k-2)>6k-6$. It now follows from
 \autoref{few1} and \autoref{few2} that there is at least one vertex in $z \in
 N_G(v)-V_1'$ and an edge $e=vz$ such that $X_e$ does not violate any
 of conditions \textbf{(i)}-\textbf{(iii)}, a contradiction.

Observe that the second part of the proof yields an algorithm for
computing a completable pair of classic $(r,k)$-certificates from
$(X_1,X_2)$. Every time we try to add an edge $e$ to $X'_i$, we test
if $(X_e,X'_{3-i})$ satisfies conditions \textbf{(i)}-\textbf{(iii)}. Conditions \textbf{(i)}-\textbf{(ii)} can clearly be checked in polynomial time and, by
\autoref{uncheck}, condition \textbf{(iii)} can also be checked in
polynomial time.  Never checking an edge that is parallel to one we have
already checked, after at most $6k-6$ failed attempts, we manage to
add a new vertex to $V_i'$. We repeat this procedure $|V_i^*|\leq
2k-2$ times. It follows that a completable pair of classic
$(r,k)$-certificates can be computed in time $k^2\cdot n^{O(1)}= n^{O(1)}$. By the
second part of \autoref{unclass}, we can then find  two edge-disjoint
$(r,k)$-safe spanning trees in $G$ in polynomial time. The overall running
time of the algorithm is polynomial, as claimed.
\end{proof}
We are now ready to proceed to the proof of \autoref{treefpt}.

\begin{proof3}
 We may suppose that there are at most two parallel edges from $u$ to
 $v$ for any $u,v \in V+r$.
If $|V|<2k-2$, the problem can be solved by a brute force algorithm in time $2^{O(k^2)}$, by generating all pairs of subgraphs of $G$ and checking whether any of these pairs satisfies the required conditions. We may hence also suppose that
$|V|\geq 2k-2$.

 Let $X_1=(V_1+r,E_1),X_2=(V_2+r,E_2)$ be
 two subgraphs of $G$. In order to test whether $X_i$ is a compact
 $(r,k)$-certificate, we first check whether $X_i$ is a tree such that
 all the vertices in $X_i \cap L_G$ are leaves of $X$. If this is the
 case, we add a set $V_i^*$ of $2k-2-|V_i|$ new vertices to $X_i$. We
 then test all possibilities to add one edge from $(V_i+r)\cap L_G$ to
 $v$ for all $v \in V_i^*$. As $|V_i+r|\leq 2k$ and $|V_i^*|\leq 2k$,
 there are at most $2k^{2k} = 2^{O(k \cdot \log k)}$ possibilities to check. For each of them, we can check in time polynomial in $k$ if the
 obtained graph is an $(r,k)$-safe spanning tree. By the definition of
 compact $(r,k)$-certificates, we can therefore check in time $2^{O(k \cdot \log k)}$ whether both $X_1$ and $X_2$ are compact
 $(r,k)$-certificates. By \autoref{uncheck}, we can therefore decide
 in time $2^{O(k \cdot \log k)} + n^{O(1)}$ if $(X_1,X_2)$ is
 a completable pair of compact $(r,k)$-certificates in $G$. By
 \autoref{compactun}, it therefore suffices to prove that there are at most
 $2^{O(k^2 \cdot \log k)}$ possible candidates for the completable pair of compact
 $(r,k)$-certificates, and that they can be generated within the same running time.

Let $X=(V'+r,E')$ be a compact $(r,k)$-certificate in $G$. Observe that
every vertex in $V'$ can be reached from $r$ by a path all of whose
interior vertices are in $S_G$ and whose length is at most $k-1$.
As every vertex in $S_G$ has at most $8k-8$
neighbors, we obtain that the number of vertices that can be reached
by such a path is at most $1+(8k-8)+(8k-8)^2+\ldots+(8k-8)^{k-1}\leq
(8k)^{k}$. As $V'$ contains at most $2k-2$ vertices, there are at most
${(8k)^{k} \choose 2k-2}\leq (8k)^{2k^2}$ possibilities to choose
$V'$.

 Now suppose that we have chosen $V'$ of size $2k-2$. Observe that
 there are at most $2{|V'+r| \choose 2}\leq 8 k^2$ edges that have
 both ends in $V'+r$. As $|E'|= 2k-2$, there are at most ${8 k^2
   \choose 2k-2}\leq (8 k^2)^{2k} $ possibilities to choose $A'$. It
 follows that there are at most $(8k)^{2k^2}\cdot (8 k^2)^{2k}$
 possibilities to choose a compact $(r,k)$-certificate $X$. As these can
 be computed by a brute force method, the algorithm can finish after
 checking less than $f(k)={(8k)^{2k^2}\cdot (8 k^2)^{2k} \choose 2}  = 2^{O(k^2 \cdot \log k)}$
 candidates for the pair of compact $(r,k)$-certificates.

If no completable pair of compact $(r,k)$-certificates
exists, by \autoref{unclass}, $G$ does not contain two edge-disjoint
$(r,k)$-safe spanning trees. On the other hand, once we have found a pair of
completable compact $(r,k)$-certificates, we can compute in polynomial time the two
edge-disjoint $(r,k)$-safe spanning trees by the second part of \autoref{unclass}. The overall running time of the obtained algorithm is $2^{O(k^2 \cdot \log k)} \cdot n^{O(1)}$.
\end{proof3}

\section{A hardness result for packing $(r,k)$-safe
  spanning trees}
\label{sec:proof-hard-trees}

In this section we prove
  \autoref{theo:lower-graph}. It is well-known that the {\sc 3-Sat} problem is \NP-complete. Further, we will need the following lemma derived from the \ETH using the so-called Sparsification
  Lemma~\cite{I}.

\begin{lemma}[Impagliazzo et al.~\cite{I}]
  \label{lem:ETH}
Assuming the \ETH,  there is an $\varepsilon > 0$ such that there is no algorithm for solving a {\sc 3-Sat} formula
with $\ell$ variables and $m$ clauses in time $2^{\varepsilon m} \cdot (\ell+m)^{O(1)}$.
\end{lemma}

The proof of \autoref{theo:lower-graph} given below
  is strongly inspired from the reduction given
  in~\cite[Theorem~5.2]{bhy}, but we provide it here entirely for the sake of
  completeness.


\begin{proof4}
Observe that, given a rooted graph $G=(V+r,E)$ and two positive
integers $p$ and $k$, $G$ contains an $(r,k)$-safe spanning tree if
and only if the graph that is obtained from $G$ by replacing each of
its edges by $p$ parallel copies of itself contains $p$ edge-disjoint
$(r,k)$-safe spanning trees. Hence, it suffices to prove the statement
for $p = 1$. Let $\phi$ be an instance of {\sc 3-Sat}, with variables
$x_1,x_2, \dots ,x_{\ell}$ and clauses $C_1, C_2, \dots , C_m$.
Adding a variable that is not contained in any clause if necessary, we
can assume that $\ell$ is even. We construct a simple rooted graph
$G=(V+r,E)$ as follows; see \autoref{fig:r-rooted-graph} for an
illustration. For $i=1,\dots ,\ell$ let $V_i$ be an independent set
containing two vertices $v_i$ and $\bar{v_i}$, and let $r$ and $t$ be
two extra vertices. Add all possible edges between $r$ and $V_1$,
between $V_i$ and $V_{i+1}$ for $i=1,\dots ,\ell-1$, and between
$V_{\ell}$ and $t$. Next, add $m$ vertices $c_1,\dots ,c_m$ and link
$c_i$ to $v_j$ (resp. $\bar{v_j}$) with a path containing $\ell/2$ interior vertices if
$x_i$ (resp. $\bar{x_i}$) is a literal of $C_i$. Finally, let
$k:=1+\ell+3\ell m/2+m$ and add to $G$ an independent set $Q$ on $q$
vertices all linked to $t$, where $q >k-\ell-1$ will be specified later. Notice
that we have $n=|V|=q+k+\ell+1$.

\begin{figure}[!ht]
\centering
\scalebox{0.73}{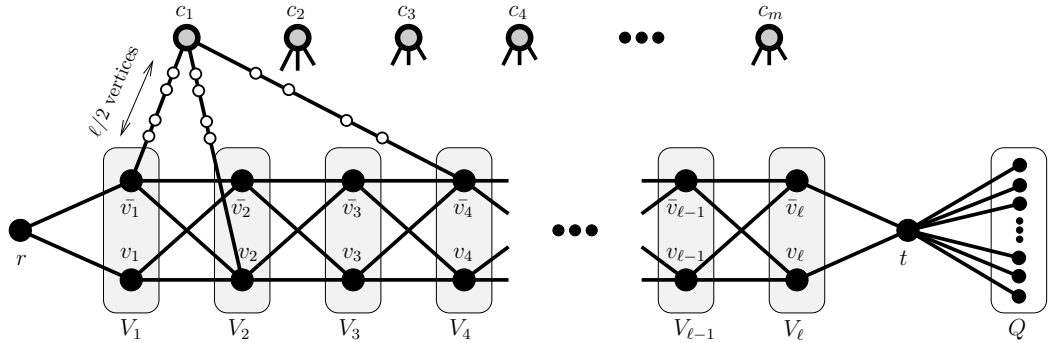}
\caption{The rooted graph $G$ in the proof of
  \autoref{theo:lower-graph} with $C_1=\bar{x_1}\vee x_2
\vee \bar{x_4}$.}
\label{fig:r-rooted-graph}
\end{figure}

We now prove that $\phi$ is satisfiable if and only if $G$ admits an
$(r,k)$-safe spanning tree. First assume that $G$ contains an
$(r,k)$-safe spanning tree $T$. Observe that, by construction and
the definition of $(r,k)$-safe spanning tree, $T-r$ has exactly two
 components whose vertex sets we denote by $T_1$ and
$T_2$. We may assume that $T_1$ contains $t$. Then $G[T_1]$ contains a
path $P$ from $V_1$ to $t$ and $T_1$ contains $Q$. As a shortest path
from $V_1$ to $t$ contains $\ell+1$ vertices, we have $|T_1|\ge
q+\ell+1$ and $n-|T_1|\le k$.  As $T$ is $(r,k)$-safe, we obtain that
$|T_1|=q+\ell+1$ and $T_1$ consists exactly of the vertex set of the
path $P$, which contains exactly $\ell+1$ vertices, and the whole set $Q$. In
particular, $P$ intersects each $V_i$ in exactly one vertex. Now for
every variable $x_i$ of $\cal F$, we set $x_i$ to \true if $P$
contains $\bar{v_i}$ and to \false if $P$ contains $v_i$. For each
clause $C_j$, there must be a path from the corresponding vertex $c_j$ to $V_1$ in $T_2$.  Then one of the vertices corresponding to a
literal of $C_j$ must not be contained in $P$ and so, this literal is
set to \true and $C_j$ is satisfied by it. It follows that the
constructed assignment satisfies $\phi$.

  Conversely, assume that $\phi$ admits a truth assignment.  Let $P$
  be the path defined so that, for every $i=1,\dots ,\ell$, it
  contains $v_i$ if $x_i$ is set to \false and $\bar{v_i}$ if $x_i$ is
  set to \true. Further, let $T_1 = V (P) \cup t \cup Q$. As $\phi$ is
  satisfied, for every $1\le j\le m$, there exists a path from $c_j$
  to $V_1\cup \cdots \cup V_{\ell}$ in $G- (T_1\cup r)$. It follows that $G-
  (T_1\cup r)$ is connected and we select a spanning tree of it. The
  union of this spanning tree with $G[T_1]$, $r$, and the edges
  incident to $r$ is a spanning tree of $G$. In order to see that this
  spanning tree is $(r,k)$-safe, it suffices to observe that, as $q>k-\ell-1$, we have
  $|T_1|=q+\ell+1>k$ and $|V-T_1|=q+k+1+\ell-(q+\ell+1)=k$.

If we fix $q=k$, then the size of $G$ is bounded by a polynomial in
$\ell$ and~$m$. Thus, the above reduction implies that, as {\sc 3-Sat} is \NP-complete, given a rooted
graph $G$ and an integer $k$, deciding whether $G$ admits an $(r,k)$-safe
spanning tree is \NP-complete.

Let now $\varepsilon$ be a positive constant and assume that $k$ is an
integer function satisfying $(\log(n))^{2+\varepsilon}\leq k(n) \leq
\frac{n}{2}$ for all $n>0$. Furthermore, suppose that there exists a
constant $C^*$ such that for all $c \geq C^*$ there exists an $n$ such
that $k(n)=c$. Finally, for the sake of a contradiction assume that
there exists a polynomial-time algorithm {\bf A}, running in time
$O(n^{c_0})$ for some $c_0>0$, for deciding if a given rooted graph
$G=(V+r,E)$ on $n$ vertices contains an $(r,k(n))$-safe spanning
tree. Then let $\phi$ be a {\sc 3-Sat} formula with $\ell$ variables
and $m$ clauses. We may assume that $\ell$ and $m$ are large enough so
that $1+\ell+3\ell m/2+m\ge C^*$.  Adding trivial clauses if
necessary, we may also assume that $\ell\le m$. By hypothesis, there
exists $n$ such that $k(n)=1+\ell+3\ell m/2+m$. So, in the above
reduction, we choose $q$ to be $n-(k(n)+\ell +1)$, in order to have
$n=q+k(n)+\ell+1$. Then, using algorithm {\bf A}, one could decide if
$\phi$ is satisfiable in time $O(n^{c_0})=O(2^{c_0 \cdot \log n})=
O(2^{c_0 \cdot k(n)^{1/(2+\varepsilon)}})$, where we have used the
hypothesis that $k(n) \geq (\log(n))^{2+\varepsilon}$.  Moreover, in
the previous construction, we have $k(n)=1+\ell+3\ell m/2+m\le
3m+3m^2/2\le 3m^2$. So we could decide whether $\phi$ is satisfiable
in time $O(2^{c_0 \cdot (3m^2)^{1/(2+\varepsilon)}})=O(2^{c_0' \cdot
  m^{\varepsilon'}})$ with $\varepsilon'=2/(2+\varepsilon)<1$, a
contradiction to \autoref{lem:ETH} assuming the \ETH.
\end{proof4}

\section{Conclusion}\label{con}

We considered three problems on finding certain disjoint substructures in graphs and digraphs. While in our proofs we restrict to finding two
of these substructures for the sake of simplicity, our results can be
generalized to allow for finding an arbitrary number of them using the
same proof techniques. More concretely, the following results can be
established using the techniques of this article. As in~\cite{bhy}, we omit the proofs of these generalized statements.
\begin{theorem}\label{fyg}
Given a rooted digraph $D=(V+r,A)$ and an integer $p \geq 2$,
deciding whether $D$ contains $p$ arc-disjoint $k$-safe spanning
$r$-arborescences is \FPT with parameter $k$. More precisely, the problem can be solved in time $2^{O(p \cdot k^2 \cdot \log k)} \cdot n^c$, where $c$ is a constant depending on $p$.
Further, if they exist, the
$p$ arc-disjoint $k$-safe spanning $r$-arborescences can be computed
within the same running time.
\end{theorem}
\begin{theorem}
 Given a rooted digraph $D=(V+r,A)$ and an integer $p\geq 2$, deciding whether $D$ contains $p$ arc-disjoint $(r,k)$-flow branchings is \FPT with  parameter $k$. More precisely, the problem can be solved in time $2^{O(p \cdot k^2 \cdot \log k)} \cdot n^c$, where $c$ is a constant depending on $p$. Further, if they exist, the $p$ arc-disjoint $(r,k)$-flow branchings can be computed within the same running time.
\end{theorem}
\begin{theorem}\label{thm:last}
Given a rooted graph $G=(V+r,E)$ and an integer $p \geq 2$, deciding whether $G$ contains $p$ arc-disjoint $(r,k)$-safe spanning trees is \FPT with parameter $k$. More precisely, the problem can be solved in time $2^{O(p \cdot k^2 \cdot \log k)} \cdot n^c$, where $c$ is a constant depending on $p$. Further, if they exist, the $p$ edge-disjoint $(r,k)$-safe spanning trees can be computed within the same running time.
\end{theorem}


It is natural to ask whether the dependency on $k$ of our \FPT algorithms can be improved. In the case of $k$-safe spanning $r$-arborescences (cf.~\autoref{safefpt}), we can derive a lower bound from \autoref{theo:lower-arbo}. Indeed, a
  corollary of \autoref{theo:lower-arbo} is that, assuming the \ETH, for any two constants $\varepsilon >0$ and
  $c>0$,
  deciding whether a rooted digraph contains two arc-disjoint $k$-safe spanning
  arborescences cannot be solved in time $2^{c \cdot k^{1-\varepsilon}} \cdot n^{O(1)}$.
  To see this, note that if such an algorithm existed, letting $k(n):=(\log(n))^{1+\varepsilon}$ we would obtain an algorithm in time $2^{c \cdot (\log(n))^{(1+\varepsilon)(1-\varepsilon)}} \cdot n^{O(1)} = n^{O(1)}$, contradicting \autoref{theo:lower-arbo}. In other words, assuming the \ETH, the problem cannot be solved in time $2^{O(k^{1-\varepsilon})} \cdot n^{O(1)}$ for any $\varepsilon>0$.


Similarly to \autoref{theo:lower-arbo},
  \autoref{theo:lower-branching}  implies a lower bound
   for packing $(r,k)$-flow branchings (cf.~\autoref{safefpt}): assuming the \ETH, deciding whether a  rooted
  digraph contains two arc-disjoint $(r,k)$-flow branchings
  cannot be solved in time $2^{O( k^{1-\varepsilon})} \cdot n^{O(1)}$ for any $\varepsilon >0$. Also, concerning $(r,k)$-safe spanning
  trees (cf.~\autoref{treefpt}), a consequence of \autoref{theo:lower-graph} is that, assuming the \ETH, for every  $p\geq 1$, deciding whether a
  rooted graph contains $p$ edge-disjoint $(r,k)$-safe spanning
  trees cannot be solved in time
  $2^{O(k^{1/2-\varepsilon})} \cdot n^{O(1)}$ for any $\varepsilon >0$.

  There is still a significant gap between the above lower bounds, which are $2^{O(k^{1-\varepsilon})}$ or $2^{O(k^{1/2-\varepsilon})}$, and the function
$2^{O(k^2 \cdot \log k)}$ in our \FPT algorithms.

We did not focus on optimizing the polynomial factors in $n$ of our algorithms, and we leave it for further research. Further, we leave as an open question whether any of the considered problems admits a polynomial kernel parameterized by $k$. Finally, it would be interesting to find a theorem on packing $k$-safe mixed arborescences in mixed graphs, hence generalizing both \autoref{fyg} and \autoref{thm:last}.

\bibliography{Bib}


\end{document}